\pdfoutput=1

\documentclass[11pt]{article}

\usepackage{times}
\usepackage{amsmath}
\usepackage{amsfonts}
\usepackage{amssymb}
\usepackage{amsthm}
\usepackage{graphicx}
\usepackage[margin=1in]{geometry}
\usepackage{paralist}
\usepackage{xspace}
\usepackage[linesnumbered,algoruled,
vlined]{algorithm2e}
\SetKwInOut{Input}{input}
\SetKwInOut{Output}{output}
\SetKwComment{Comment}{// }{}
\SetCommentSty{textit}
\SetEndCharOfAlgoLine{}
\def\DEF{\stackrel{\rm def}{=}}

\usepackage{color}
\usepackage{soul}

\newtheorem{theorem}{Theorem}[section]
\newtheorem{lemma}[theorem]{Lemma}
\newtheorem{definition}[theorem]{Definition}
\newtheorem{corollary}[theorem]{Corollary}

\newcommand{\namedref}[2]{\hyperref[#2]{#1~\ref*{#2}}}
\newcommand{\sectionref}[1]{\namedref{Section}{#1}}

\newcommand{\theoremref}[1]{\namedref{Theorem}{#1}}
\newcommand{\defref}[1]{\namedref{Definition}{#1}}
\newcommand{\figureref}[1]{\namedref{Figure}{#1}}

\newcommand{\lemmaref}[1]{\namedref{Lemma}{#1}}

\newcommand{\corollaryref}[1]{\namedref{Corollary}{#1}}

\newcommand{\equalityref}[1]{\hyperref[#1]{Equality~\eqref{#1}}}
\newcommand{\inequalityref}[1]{\hyperref[#1]{Inequality~\eqref{#1}}}

\newcommand{\N}{\mathbb{N}}
\newcommand{\R}{\mathbb{R}}
\newcommand{\BO}{\mathcal{O}}
\newcommand{\SPD}{\mbox{SPD}}
\newcommand{\Congest}{\textsc{congest}\xspace}
\def\CONGEST{\Congest}
\DeclareMathOperator{\Wd}{wd}
\DeclareMathOperator{\WD}{WD}
\newcommand{\Wdp}{\Wd'}
\DeclareMathOperator{\Hd}{hd}
\DeclareMathOperator{\HD}{D}

\newcommand{\Paths}{\mathrm{paths}}
\newcommand{\Set}[1]{\left\{ #1 \right\}}

\newcommand{\Next}{\mathrm{next}}
\newcommand{\rtc}{\textsc{rtc}}

\usepackage[colorlinks,linkcolor=blue,filecolor=blue,citecolor=blue,urlcolor=blue,pdfstartview=FitH,plainpages=false]{hyperref}

\def\eps{\varepsilon}

\begin{document}

\title{Fast Partial Distance Estimation and Applications}

\author{Boaz Patt-Shamir, Tel Aviv University\\
Christoph Lenzen, Max-Planck Institute for Informatics}

\def\thepage{}
\begin{titlepage}
\maketitle

\begin{abstract}
We study approximate distributed solutions to the weighted
\emph{all-pairs-shortest-paths} (APSP) problem in the \Congest model. We obtain
the following results.
\begin{compactenum}
\item A deterministic $(1+o(1))$-approximation to APSP in $\tilde{\BO}(n)$
rounds. This improves over the best previously known algorithm, by
both derandomizing it and by reducing the running time by a $\Theta(\log
n)$ factor.
\label{enum1}
\end{compactenum}
In many cases, routing schemes involve relabeling, i.e., assigning new
names to nodes and require that these names are used in distance and
routing queries. 
It is known that relabeling  is necessary to achieve running times of
$o(n/\log n)$.  In the relabeling model, we obtain the
following results.
\begin{compactenum}
\setcounter{enumi}{1}
\item 
  A randomized $\BO(k)$-approximation to APSP, for any integer $k>1$, running in
$\tilde{\BO}(n^{1/2+1/k}+D)$ rounds, where $D$ is the hop diameter of the
network. This algorithm simplifies the best previously known result and reduces
its approximation ratio from $\BO(k\log k)$ to $\BO(k)$. Also, the new
algorithm uses
uses labels of asymptotically optimal size, namely $\BO(\log n)$ bits.
\item 
  A randomized $\BO(k)$-approximation to APSP, for any integer $k>1$,
  running in time $\tilde{\BO}((nD)^{1/2}\cdot n^{1/k}+D)$ and producing
  \emph{compact routing tables} of size $\tilde{\BO}(n^{1/k})$. The node lables
  consist of $\BO(k\log n)$ bits. This improves on the approximation ratio of
  $\Theta(k^2)$ for tables of that size achieved by the best previously known
  algorithm, which terminates faster, in $\tilde{\BO}(n^{1/2+1/k}+D)$ rounds.
\end{compactenum}
\end{abstract}
\end{titlepage}
\pagenumbering{arabic}

\section{Introduction}
To allow a network to be useful, it must facilitate routing
messages between nodes. By the very nature of networks, this
computation must be distributed, i.e., there must be a distributed
algorithm that computes the local data structures that support routing
at network junctions (i.e., routing tables at nodes). A trivial
distributed algorithm for 
this purpose is to collect the entire topology at a single location,
apply a centralized algorithm, and distribute the result via the
network. This simplistic approach is costly, in particular if the
available bandwidth is 
limited. To study the distributed time complexity of routing table
computation, we use the view point of
the \Congest model, i.e., we assume that each link in an $n$-node
network allows only for $\BO(\log n)$ bits to be
exchanged in each unit of time.

In this work, we consider networks modeled by weighted undirected
graphs, where the edge weights represent some abstract link cost,
e.g., latency. 
%
Regarding  routing, which is is a generic task with
many variants, we focus on the following specific problems.
\begin{compactitem}
\item \emph{All-pairs distance estimation}: How fast can each node obtain an estimate of the distance to
each other node, and how good is that estimate?
\item \emph{All-Pairs Shortest Paths:} How fast can we construct local
  data structures so that when given a destination node identifier,
  the node can locally determine the 
  next hop on 
a  path to the destination, and what's the \emph{stretch} of the
resulting route w.r.t.\ the shortest path? 
\end{compactitem}
The variants above are graph-theoretic; in modern routing systems, it
is common practice to assign to nodes labels (identifiers) that
contain some routing information. IP addresses, for example, contain a
``network'' part and a ``host'' part, which allow for hierarchical
routing. Thus,  the following
additional questions are also of interest to us.
\begin{compactitem}
\item \emph{Routing Table Construction:} What are the answers to the above
questions if we allow 
\emph{relabeling}, i.e., allow the algorithm to choose (small)
\emph{labels} as node identifier, and require
that distance and routing queries
refer to nodes using these labels?
\item \emph{Compact Routing:} What are the answers to the above
  questions  if we require that
  storage space allocated at the nodes (routing table size) is small?
\end{compactitem}

\paragraph{Background.}

Shortest paths are a central object of study since the dawn of the 
computer era. The Bellman-Ford algorithm \cite{Bellman,Ford-56}, 
although originally developed for centralized optimization, is one of
the few most 
fundamental distributed algorithms. Implemented as RIP, the algorithm 
was used in the early days of the Internet (when it was still called ARPANET) 
\cite{MQ-77}. Measured in terms of the \Congest model, a Bellman-Ford 
all-pairs shortest paths computation in weighted graphs takes 
$\Theta(n^2)$ time in the worst case, and requires $\Theta(n\log n)$
bits of storage 
at each node. Another simple solution to the problem is to collect the 
complete topology at each node (by flooding) and then apply a local 
single-source shortest paths algorithm, such as Dijkstra's. This 
solution has time complexity $\Theta(m)$ and storage complexity 
$\Theta(m)$, where $m$ denotes the number of links in the network. It 
also enjoys improved stability and flexibility, and due to these 
reasons it became the Internet's routing algorithm in its later stages
of  evolution
(see discussion in \cite{ARPANET}). It is standardized as
OSPF~\cite{ospf}, which contains, in addition, provisions for
hierarchical routing.

\noindent\textbf{State of the art and New Results.}
%
Recently there has been a flurry of new results about routing in the
\Congest model. Instead of trying to track them all, let us start by
reviewing known lower bounds, which help placing our
results in the context of what is possible. 
\begin{compactitem}
\item 
Without relabeling, any polylogarithmic-ratio approximation to APSP requires
$\tilde{\Omega}(n)$ rounds~\cite{nanongkai14,lenzen12}.%
\footnote{
Throughout this paper, we use $\tilde{\BO}$-notation, which hides
poly-logarithmic factors. See \sectionref{sec-general}.
}
 This holds also if
tables must only enable either distance estimates or routing, but not both.
\item With node relabeling, any approximation to APSP requires
  $\tilde{\Omega}(\sqrt{n}+D)$ rounds~\cite{dassarma12hardness}, where
  $D$ denotes the \emph{hop diameter} of the network (see \sectionref{sec-graph}. The
  bound holds for both 
  routing and distance queries, and even for $D\in\BO(\log n)$. (However, if routing may be
  \emph{stateful}, i.e., routing decisions may depend on the tables of
  previously visited nodes, no non-trivial lower bound is known; all our routing
  algorithms are stateless.)
\item If the routing table size is $\tilde{\BO}(n^{1/k})$, then the
  approximation ratio of the induced routes is
at least $2k-1$~\cite{abraham08,peleg89trade-off}. (This result does
not hold for stateful routing.) For distance approximation, the same
bound  has been established for the special cases of
$k\in\{1,2,3,5\}$, and is conjectured to hold for any $k$ (see \cite{zwick01} and
references therein).
\item  It is known \cite{holzer14} that any
  randomized $(2-o(1))$-approximation of APSP, and that any
  $(2-o(1))$-approximation of the weighted diameter of a graph takes
  $\tilde\Omega(n)$ time
  in the worst case. 
In the \emph{unweighted}
case, \cite{FHW-12} show an
$\tilde\Omega(n)$ lower bound on the time required to approximate the
diameter to within a $3/2$ factor.
\end{compactitem}

\noindent
Let us now review the best known upper bounds and compare them with
our results.
\begin{compactitem}
\item For any $\varepsilon>0$, we give a \emph{deterministic}
$(1+\varepsilon)$-approximation to APSP that runs in $\BO(\varepsilon^{-2}n \log
n)$ rounds. The best known previous result, due to Nanongkai~\cite{nanongkai14},
achieves the same approximation ratio within $\BO(\varepsilon^{-2}n \log^2 n)$
rounds with high probability---Nanongkai's algorithm is
randomized. We note that independently and concurrently to our work, Holzer and
Pinsker~\cite{holzer14} derived the same algorithm and result, and applied it in
the Broadcast Congested Clique model, in which in each round, each node 
posts a single $\BO(\log n)$-bit message which is delivered to all
other nodes.
\item For any $k\in \N$, we obtain a randomized $(6k-1+o(1))$-approximation to
APSP running in time $\tilde{\BO}(n^{1/2+1/(4k)}+D)$. 
The algorithm succeeds with
high probability (see \sectionref{sec-general}), as do all our
randomized algorithms. This
improves our previous work~\cite{lenzen12} by simplifying it and by reducing
the approximation ratio from $\BO(k\log k)$ to $\BO(k)$. The new
algorithm relabel nodes with labels of $\BO(\log n)$ bits, whereas the
previous one required $\BO(\log n\log k)$-bit labels.
\item For any $k\in \N$, we give a randomized $(4k-3+o(1))$-approximation to
APSP running in time $\tilde{\BO}(\min\{(nD)^{1/2}\cdot
n^{1/k},n^{2/3+2/(3k)}+D)$ with tables of size $\tilde{\BO}(n^{1/k})$. This
improves over the stretch of $\BO(k^2)$ in own previous work~\cite{lenzen12}, at
the cost of increasing the running time (from $\tilde{\BO}(n^{1/2+1/k}+D)$). We
point out, however, that the proposed algorithm is the first that achieves an
asymptotically optimal trade-off between stretch and table size in time
$\tilde{o}(n)$ for all graphs of diameter $D\in \tilde{o}(n)$ and $k>2$.
\end{compactitem}

\noindent\textbf{Technical Discussion.} Our key tool is a generalization of the
\emph{$(S,h,\sigma)$-detection problem}, introduced in
\cite{lenzen13}.\footnote{In the 
original paper, the third parameter is called $k$. We use $\sigma$ here to avoid
confusion with the use of $k$ as the parameter controlling the trade-off between
approximation ratio and table size.} 
The problem is defined as follows. Given a graph  with a distinguished set of
\emph{source nodes} $S$, the task is for each node to find the distances to its
closest $\sigma\in \N$ sources within $h\in \N$ hops (the formal
definition is given in \sectionref{sec-pde-def}). 
 In~\cite{lenzen13} it is shown that 
this task can be solved in $h+\sigma$ rounds on \emph{unweighted}
graphs. The main
new ingredient in all our results is an algorithm that, within a comparable
running time, produces an \emph{approximate}  solution to 
$(S,h,\sigma)$-detection in \emph{weighted} graphs.

Weighted graphs present a significant difficulty, because in weighted graphs, the
number of hops in a shortest (by weight) path between two nodes may be a factor
of $\Theta(n)$ larger than the minimal number of hops on any path connecting the
same two nodes (the Congested Clique provides an extreme example of this
phenomenon). Therefore, naively finding the absolute closest $\sigma$ sources
(w.r.t.\ weighted distance) within $h$ hops may require running time $\Omega(n)$
in the worst case, for any  $h$ and $\sigma$. One may circumvent this difficulty
by replacing the underlying graph metric by \emph{$h$-hop distances}, which for
$v,w\in V$ is defined as the minimum weight of all $v$-$w$-paths that consist of
at most $h$ hops. The collection of $h$-hop distances does not constitute a
metric, but one can solve the $(S,h,\sigma)$-detection problem under $h$-hop
distances in time $\sigma h$ using techniques similar to those used in the
unweighted case~\cite{lenzen12}.

\begin{figure}[t]
\begin{center}
\includegraphics[width=\columnwidth]{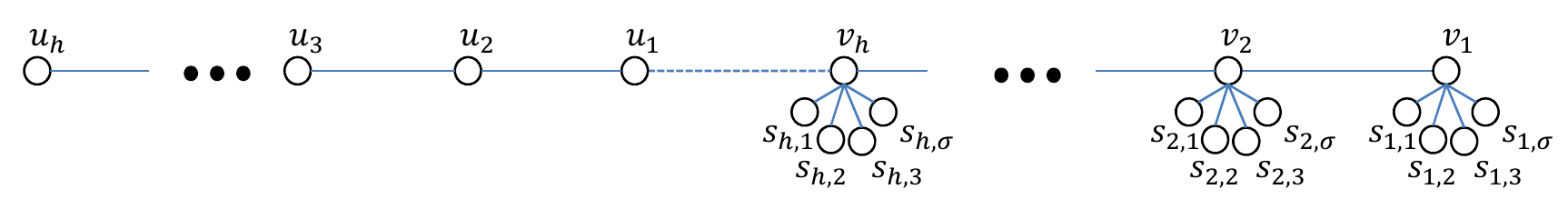}
\end{center}
\caption{\small A graph where $(S,h+1,\sigma)$-detection cannot be solved in
$o(h\sigma)$ rounds. Edge weights are $4ih$ for edges $\{v_i,s_{i,j}\}$ for all
$i\in \{1,\ldots,h\}$ and $j\in \{1,\ldots,\sigma\}$, and $1$ (i.e.,
negligible) for all other edges. Node $u_i$, $i\in \{1,\ldots,h\}$, needs to
learn about all nodes $s_{i,j}$ and distances $\Wd_{h+1}(u_i,s_{i,j})$, where
$j\in \{1,\ldots,\sigma\}$. Hence all this information must traverse the dashed
edge $\{u_1,v_h\}$. The example can be modified to attach the same source set
to each $v_h$. Varying distances, then still $\sigma h=|S| h$ values must be
communicated over the dashed edge. Hence, the special case $|S|=\sigma$ is not
easier.}\label{fig:source_weighted_lower}
\end{figure}

Unfortunately, as illustrated in \figureref{fig:source_weighted_lower}, this
time complexity is optimal in the worst case. To avoid this bottleneck,
Nanongkai~\cite{nanongkai14} used a rounding technique that had previously been
leveraged in the centralized setting~\cite{zwick02}, solving the
problem to within $(1+\eps)$ factor by
essentially reducing the weighted instance to $\BO(\log n /
\varepsilon)$ unweighted instances. In~\cite{nanongkai14}, the idea is
to solve each instance using  breadth-first-search. To avoid
collisions, 
a random delay is applied to the starting time of each instance. The
result is  $(1+\varepsilon)$-approximate distances to all sources in
$\BO((h+|S|)\log^2 n /\varepsilon^2)$ rounds w.h.p. 
We replace this
part of the 
algorithm with the deterministic source detection algorithm
from~\cite{lenzen13}, obtaining a deterministic algorithm of running
time $\BO((h+|S|)\log n)$. This, using $S=V$ and $h=\sigma=n$, has the
immediate corollary of  a 
deterministic $(1+o(1))$-approximation to APSP.

To derive our other results, solving the special case of $\sigma=|S|$ is
insufficient. Consequently, we define  a
$(1+\eps)$-approximate version of the $(S,h,\sigma)$-detection problem
which we call \emph{partial distance estimation}, or \emph{PDE} for
short (see \defref{def-pde}).
The crucial insight is that by combining Nanongkai's and Zwick's
rounding scheme with the algorithm from~\cite{lenzen13}, PDE can be solved
in  $\BO((h+\sigma)\log n /\varepsilon^2)$ rounds, with a bound
of  $\BO(\sigma^2)$ on the number of messages sent by each node
throughout the computation. Exploiting
these properties carefully, we obtain our other results.

\emph{Further ingredients.} Our compact routing schemes are distributed
constructions of the routing hierarchies of Thorup and Zwick~\cite{thorup01}.
These make use of efficient tree labeling schemes presented in the same paper,
which allow for a distributed implementation in time $\tilde{\BO}(h)$ in trees
of depth $h$ if relabeling is permitted. For compact routing table construction,
we continue the Thorup-Zwick construction by simulating the partial distance
estimation algorithm on the \emph{skeleton graph}~\cite{lenzen12}, broadcasting
all messages via a BFS tree. This avoids the quadratic stretch incurred by the
approach in~\cite{lenzen12} due to approximating distances in the skeleton graph
using a \emph{spanner}~\cite{peleg89}, which is constructed by simulating the
Baswana-Sen algorithm~\cite{baswana07}. If compact tables are not required, the
partial distance estimation algorithm enables to collapse the Thorup-Zwick
hierarchy of the lower levels into a single step, giving rise to a
constant 
approximation ratio, and thus removing the multiplicative overhead of
$\BO(\log k)$ 
from~\cite{lenzen12} in constructing non-compact routing tables.

\noindent\textbf{Organization of the paper.} In \sectionref{sec:model} we define
the basic graph-theoretic concepts, problems, and execution model. In
\sectionref{sec:pde} we describe the algorithm for partial distance estimation.
In \sectionref{sec:apps} we present the improved results we derive using the
algorithm for PDE.

\section{Model and Problem}
\label{sec:model}

\subsection{Execution Model}
We follow the $\CONGEST$ model as described by \cite{Peleg:book}. The
distributed system is represented by a simple, connected weighted graph
$G=(V,E,W)$, where $V$ is the set of nodes, $E$ is the set of edges, and $W:E\to
\N$ is the edge weight function. As a convention, we
use $n$ to denote the number of nodes. We assume that all edge weights are
bounded by some polynomial in $n$, and that each node $v\in V$ has a unique
identifier of $\BO(\log n)$ bits (we use $v$ to denote both the node and its
identifier).

Execution proceeds in global synchronous rounds, where
in each round, each node takes the following three steps:
\begin{compactenum}
\item Perform local computation,
\item send messages to neighbors, and
\item receive the messages sent by neighbors.
\end{compactenum}
Initially, nodes are aware only of their neighbors; input values (if any) are
assumed to be fed by the environment before the first round. Throughout this
paper, we assume that node $v$ is given the weight of each edge $\{v,w\}\in E$
as input. Output values, which are computed at the end of the final round, are
placed in special output-registers. In each round, each edge can carry a message
of $B$ bits for some given parameter $B$ of the model; we assume that $B\in
\Theta(\log n)$ throughout this paper.

\subsection{Graph-Theoretic Concepts} 
\label{sec-graph}

Fix a weighted undirected graph $G=(V,E,W)$. A \emph{path} $p$ connecting
$v,w\in V$ is a sequence of nodes $\langle v=v_0,\ldots,v_k=w\rangle$ such that
for all $0\le i<k$, $\{v_i,v_{i+1}\}$ is an edge in $G$.
Let $\Paths(v,w)$ denote the set of all paths connecting nodes $v$ and $w$.  We
use the following unweighted concepts.
\begin{compactitem}
\item The \emph{hop-length} of a path $p$, denoted $\ell(p)$, 
is the number of edges in it.
\item A path $p_0$ between $v$ and $w$ is a {\em shortest unweighted path}
if its hop-length $\ell(p_0)$ is minimum among all $p\in\Paths(v,w)$.
\item The \emph{hop distance} $\Hd:V\times V\to \N_0$ is defined as 
the hop-length of a shortest unweighted path,
$\Hd(v,w):=\min\{\ell(p)\mid{p\in\Paths(v,w)}\}$.
\item The \emph{hop-diameter} of a graph $G=(V,E,W)$ is
$\HD\DEF\max_{v,w\in V}\{\Hd(v,w)\}$.
\end{compactitem}
We use the following weighted concepts.
\begin{compactitem}
\item The \emph{weight} of a path $p$, denoted $W(p)$, is its total
edge weight, i.e., $W(p)\DEF\sum_{i=1}^{\ell(p)} W(v_{i-1},v_i)$.
\item A path $p_0$ between $v$ and $u$ is a {\em shortest weighted path}
if its weight $W(p_0)$ is minimum among all $p\in\Paths(v,w)$.
\item The \emph{weighted distance} $\Wd:V\times V\to \N$
is defined as the weight of a shortest weighted path, 
$\Wd(v,u)\DEF\min\{W(p)\mid{p\in\Paths(v,u)}\}$.
\item The \emph{weighted diameter} of $G$ is
$\WD\DEF\max\{\Wd(v,u)\mid{v,u\in V}\}$.
\end{compactitem}

Finally, we define the notion of the \emph{shortest
paths distance}. 
\begin{compactitem}
\item A path $p_0$ between $v$ and $w$ is a {\em minimum-hop shortest weighted 
path} if its hop-length $\ell(p_0)$ is minimum among all shortest weighted 
paths connecting $v$ and $w$. The number of hops in such a path $h_{v,w}$ is
the \emph{shortest path distance} of $v$ and $w$.
\item
The {\em shortest path diameter} of $G$ is
$\SPD\DEF\max\Set{h_{v,w}\mid v,w\in V}$. 
\end{compactitem}

\subsection{Routing}

In the \emph{routing table construction} problem (abbreviated $\rtc$), the
output at each node $v$ consists of (i) a unique \emph{label} $\lambda(v)$ and
(ii) a function ``$\Next_v$'' that takes a destination label $\lambda$ and
produces a neighbor of $v$, such that given the label $\lambda(w)$ of any node
$w$, and starting from any node $v$, it is possible to reach $w$ from $v$ 
by following the $\Next$ pointers. Formally, the requirement is as follows.
Given a start node $v$ and a destination label $\lambda(w)$, let $v_0=v$ and define
$v_{i+1}=\Next_{v_i}(\lambda(w))$ for $i\ge 0$. Then $v_i=w$ for some $i$.

The performance of a solution is measured in terms of its \emph{stretch}. 
A route is said to have stretch $\rho\ge1$ if its total weight is $\rho$ times 
the weighted distance between its endpoints, and a solution to $\rtc$ is said to
have stretch $\rho$ if the maximal stretch of all its induced routes is $\rho$.

\noindent\textbf{Variants.} Routing appears in many incarnations. We list a few
important variants below.

\emph{Name-independent routing.} Our definition of $\rtc$ allows for node
relabeling. This is the case, as mentioned above, in the Internet. The case
where no such relabeling is allowed (which can be formalized by requiring
$\lambda$ to be the identity function), was introduced in \cite{ABLP89} as 
\emph{name-independent} routing. 
Note that any scheme can be trivially transformed into a
name-independent one by announcing all node/label pairs, but this naive
approach requires to broadcast and store $\Omega(n\log n)$ bits.

\emph{Stateful routing.} The routing problem as defined above is
\emph{stateless} in the sense that routing a packet is done regardless of the
path it traversed so far. One may also consider \emph{stateful} routing, where
while being routed, a packet may gather information that helps it navigate
later. (One embodiment of this idea in the Internet routing today is MPLS, where
packets are temporarily piggybacked with extra headers. Early solutions to the
compact routing problem, such as that of \cite{PU89}, were also based on
stateful routing.) Note that the set of routes to a single destination in
stateless routing must constitute a tree, whereas in stateful routing even a
single route may contain a cycle. Formally, in stateful routing the label of the
destination may change from one node to another: The $\Next_v$ function outputs
both the next hop (a neighbor node), and a new label $\lambda_v$ used in the
next hop.

\subsection{Distance Approximation}\label{sec:dist_def}

The \emph{distance approximation} problem is closely related to the routing
problem. Again, each node $v$ outputs a label $\lambda(v)$, but now, $v$ needs
to construct a function $\mathrm{dist}_v: \lambda(V)\to \R^+$ (the table) such
that for all $w\in V$ it holds that $\mathrm{dist}_v(\lambda(w))\geq \Wd(v,w)$.
The stretch of the approximation $v$ has of $\Wd(v,w)$ is
$\mathrm{dist}_v(\lambda(w))/\Wd(v,w)$, and the solution has stretch $\rho\geq
1$ if $\max_{v,w\in V}\{\mathrm{dist}_v(\lambda(w))/\Wd(v,w)\}= \rho$.

Similarly to routing, we call a scheme name-independent if $\lambda$
is the identity function. Since we require distance estimates to be
produced without communication, there is no ``stateful'' distance
approximation.

\subsection{Partial Distance Estimation (PDE)}
\label{sec-pde-def}
The basic problem we attack in this paper is partial distance
estimation, which generalizes the source detection problem. Let us
start by defining the simpler  variant.

Given a set of nodes $S\subseteq V$ and a parameter $h\in\N$, $L_v^{(h)}$
denotes the list resulting from ordering the set $\{(\Wd(v,w),w)\,|\,w\in
S\wedge h_{v,w}\leq h\}$ lexicographically in ascending order, i.e.,
\begin{equation*}
(\Wd(v,w),w)<(\Wd(v,u),u)\Leftrightarrow \big((\Wd(v,w)<\Wd(v,u))\vee
(\Wd(v,w)=\Wd(v,u)\wedge w<u)\big).
\end{equation*}
\begin{definition}[$(S,h,\sigma)$-detection] 
Given are a set of \emph{sources} $S\subseteq V$ and parameters $h,\sigma\in
\N$. Each node is assumed to know $h$, $\sigma$, and whether it is in $S$ or
not. The goal is to compute at each node $v\in V$ the list $L_v$ of the top
$\sigma$ entries in $L_v^{(h)}$, or the complete $L_v^{(h)}$ if $|L_v^{(h)}|\leq
\sigma$.
\end{definition}

Relaxing this by allowing approximation to within $(1+\eps)$, we arrive at the
following definition.
\begin{definition}[Partial Distance Estimation (PDE)]
\label{def-pde}
Given $S\subseteq V$, $h,\sigma\in \N$, and $\varepsilon>0$,
\emph{$(1+\varepsilon)$-approximate $(S,h,\sigma)$-estimation} is defined as
follows. Determine a \emph{distance function} $\Wd':V\times S\to \N\cup \infty$
satisfying
\begin{compactitem}
\item $\forall v\in V,s\in S:\,\Wd'(v,s)\geq \Wd(v,s)$, where $\Wd$ is the
weighted distance from $v$ to $s$, and
\item if $h_{v,s}\leq h$, then $\Wd'(v,s)\leq (1+\varepsilon)\Wd(v,s)$.
\end{compactitem}
For each $v\in V$, sort the set $\{(\Wd'(v,s),s)\,|\,s\in S\}$ in
ascending lexicographical order. Each node $v$ needs to output the prefix $L_v$
of the sorted list consisting of the first (up to) $\sigma$ elements with
$\Wd'(v,s)< \infty$.
\end{definition}
This generalizes the source detection problem in that setting $\varepsilon=0$
and choosing $\Wd'$ as $h$-hop distances results in an exact weighted version of
the source detection problem, and, specializing further to unweighted graphs,
$h$-hop distances just become hop distances to nodes within $h$ hops.

\subsection{General Concepts}
\label{sec-general}
We use extensively ``soft'' asymptotic notation that ignores
polylogarithmic factors. Formally, $g(n)\in \tilde \BO(f(n))$ 
if and only if there exists a constant $c\in \R^+_0$ such that
$f(n)\leq g(n)\log^c n$ for all but finitely many values of $n\in \N$.
Analogously, 
\begin{compactitem}
\item $f(n)\in \tilde\Omega(g(n))$ if and only if $g(n)\in
\tilde{\BO}(f(n))$,
\item $\tilde{\Theta}(f(n))\DEF\tilde\BO(f(n))\cap
\tilde\Omega(f(n))$,
\item $g(n)\in \tilde{o}(f(n))$ if and only if for each $c\in \R^+_0$ it
holds that $\lim_{n\to \infty}g(n)\log^c(f(n))/f(n)=0$, and
\item $g(n)\in \tilde{\omega}(f(n))$ if and only if $f(n)\in \tilde{o}(g(n))$.
\end{compactitem}

To model probabilistic computation, we assume that each node has access to an
infinite string of independent unbiased random bits.  When we say that a certain
event occurs ``with high probability'' (abbreviated ``w.h.p.''), we mean that
the probability of the event not occurring can be set to be less than $1/n^c$
for any desired constant $c$, where the probability is taken over the strings of
random bits. Due to the union bound, this definition entails that any polynomial
number of events that occur w.h.p.\ also jointly occur w.h.p. We will make
frequent use of this fact throughout the paper.

\section{From Weighted to Unweighted}\label{sec:pde}
Fix $0<\varepsilon\in \BO(1)$. Using the technique of
Nanongkai~\cite{nanongkai14}, we reduce PDE to
$\BO(\log_{1+\varepsilon}\WD)$ instances of the unweighted problem as 
follows.
Let $i_{\max}\DEF\log_{1+\varepsilon}w_{\max}$, where $w_{\max}$ is the 
largest edge weight in $G$. Note that since we assume that edge weights 
are polynomial in $n$, $i_{\max}\in
\BO(\log_{1+\varepsilon} n)$. Clearly $i_{\max}$ can be determined in 
$\BO(D)$ rounds.

For $i\in \{1,\ldots,i_{\max}\}$, let $b(i)\DEF(1+\varepsilon)^i$, and
define $W_i:E\to b(i)\cdot\N$ by 
$W_i(e)\DEF b(i)\lceil 
W(e)/b(i)\rceil$,
i.e., by rounding up edge weights to integer multiples of $(1+\varepsilon)^i$.
Denote by $\Wd_i$ the resulting distance function, i.e., the distance function
of the graph $(V,E,W_i)$. Then the following crucial property holds.
\begin{lemma}[adapted from \cite{nanongkai14}]\label{lemma:round}
For all $v,w\in V$ and 
\begin{equation*}
i_{v,w}:=\max\left\{0,\left\lfloor
\log_{1+\varepsilon}\left(\frac{\varepsilon\Wd(v,w)}{h_{v,w}}\right)
\right\rfloor\right\},
\end{equation*}
it holds that
\begin{equation*}
\Wd_{i_{v,w}}(v,w)< (1+\varepsilon)\Wd(v,w)\in
\BO\left(\frac{b(i_{v,w})h_{v,w}}{\varepsilon}\right).
\end{equation*}
\end{lemma}
\begin{proof}
W.l.o.g., $i_{v,w}\neq 0$, as $b_0=1$ and hence $\Wd_0=\Wd$. The choice of
$i_{v,w}$ yields that
\begin{equation*}
\Wd_{i_{v,w}}(v,w)< \Wd(v,w)+b(i_{v,w})h_{v,w}\leq
(1+\varepsilon)\Wd(v,w).
\end{equation*}
To see the second bound, note that by definition of $i_{v,w}$ and 
$b(i_{v,w})$,
\begin{equation*}
\Wd(v,w)\leq \frac{(1+\varepsilon)b(i_{v,w})h_{v,w}}{\varepsilon}.
\end{equation*}
Due to the previous inequality and the constraint that $\varepsilon\in \BO(1)$,
the claim follows.
\end{proof}
Next, let $G_i$ be the \emph{unweighted} graph obtained by replacing 
each edge $e$ in
$(V,E,W_i)$ by a path of $W_i(e)/b(i)$ edges. Let $\Hd_i(v,w)$ denote 
the
distance (minimal number of hops) between $v$ and $w$ in $G_i$. The 
previous lemma implies that in
$G_{i_{v,w}}$, the resulting hop distance between $v$ and $w$ is not too large.
\begin{corollary}\label{coro:round}
For each $v,w\in V$, it holds that $\Hd_{i_{v,w}}(v,w)\in
\BO(h_{v,w}/\varepsilon)$.
\end{corollary}
\begin{proof}
By \lemmaref{lemma:round}, $\Wd_{i_{v,w}}(v,w)\in
\BO(b(i_{v,w})h_{v,w}/\varepsilon)$. As edge weights are scaled down by factor
$b(i_{v,w})$ in $G_{i_{v,w}}$, this implies $\Hd_{i_{v,w}}(v,w)\in
\BO(h_{v,w}/\varepsilon)$.
\end{proof}
These simple observations imply that an efficient algorithm for unweighted
source detection can be used to solve partial distance estimation at the cost of
a small increase in running time.
\begin{theorem}\label{theorem:detection}
Any deterministic algorithm for unweighted $(S,h,\sigma)$-detection with running
time $R(h,\sigma)$ can be employed to solve $(1+\varepsilon)$-approximate
$(S,h,\sigma)$-estimation in $\BO(\log_{1+\varepsilon} n\cdot R(h',\sigma)+D)$
rounds, for some $h'\in \BO(h/\varepsilon)$.
\end{theorem}
\begin{proof}
Let $\cal A$ be any deterministic algorithm for unweighted
$(S,h,\sigma)$-detection with running time $R(h,\sigma)$. We use the following
algorithm for partial distance estimation.
\begin{compactenum}
	\item Let $h'$ be such that $\Hd_{i_{v,s}}(v,s)\leq h'$ for all 
	$v\in V$ and $s\in S$ with $h_{v,s}\leq h$. By \corollaryref{coro:round}, there
	is such an $h'\in \BO(h/\varepsilon)$.
	\item For $i\in \{1,\ldots,i_{\max}\}$, solve $(S,h',\sigma)$-detection on
	$G_i$ by $\cal A$. Denote by $L_{v,i}$ the computed list.
	\item For $s\in S$, define
	\begin{equation*}
    \tilde{\Wd}(v,s)\DEF\inf \{\Hd_i(v,s)b(i)\,|\,0\le i\le
    i_{\max}\wedge (\Hd_i(v,s),s)\in L_{v,i}\}.
	\end{equation*}
	Note that if $\tilde{\Wd}(v,s)<\infty$, then $v$ can determine $s$ and
	$\tilde{\Wd}(v,s)$ from the previous step. Each node $v$ outputs the list
	$L_v$ consisting of the (up to) first $\sigma$ elements of the set
	$\{(\tilde{\Wd}(v,s),s)\,|\,\tilde{\Wd}(v,s)<\infty\}$, with respect to
	ascending lexicographical order.
\end{compactenum}
Clearly, the resulting running time is the one stated in the claim of the
theorem.

In order to show that the output is feasible, i.e., satisfies the guarantees of
partial distance estimation with approximation ratio $1+\varepsilon$, define
\begin{equation*}
\Wdp(v,s)\DEF\inf \{\Hd_i(v,s)b(i)\,|\,0\le i\le i_{\max}\wedge \Hd_i(v,s)\leq
h'\}.
\end{equation*}
We claim that the required properties are satisfied with respect to $\Wdp$
and that the list returned by $v$ is the one induced by $\Wdp$, which will
complete the proof. The claim readily follows from the following properties.
\begin{compactenum}
\item $\forall v\in V,s\in S:~\Wdp(v,s)\geq \Wd(v,s)$,
\item $\forall v\in V,s\in S:~h_{v,s}\leq h \Rightarrow \Wdp(v,s)\leq
(1+\varepsilon)\Wd(v,s)$,
\item $\forall v\in V,s\in S:~(\tilde{\Wd}(v,s),s)\geq (\Wdp(v,s),s)$, and
\item $\forall v\in V, (\tilde{\Wd}(v,s),s)\in L_v:~\tilde{\Wd}(v,s)=\Wdp(v,s)$.
\end{compactenum}
Hence, it remains to show these four properties.
\begin{compactenum}
\item By definition, $b(i)\Hd_i(v,s)=\Wd_i(v,s)\geq \Wd(v,s)$ for all $v\in V$
and $s\in S$.
\item By the first step, $h_{v,s}\leq h \Rightarrow h_{i_{v,s}}(v,s)\leq h'$.
Hence, $\Wdp(v,s)\leq b(i_{v,s})h_{i_{v,s}}(v,s)=\Wd_{i_{v,s}}(v,s)
<(1+\varepsilon)\Wd(v,s)$ by \lemmaref{lemma:round}.
\item This trivially holds, because $(\Hd_i(v,s),s)\in L_{v,i}$ implies
that $\Hd_i(v,s)\leq h'$ (we executed $(S,h',\sigma)$-detection on each $G_i$),
i.e., $\tilde{\Wd}(v,s)$ is an infimum taken over a subset of the set used for
$\Wdp(v,s)$.
\item Assume for contradiction that $(\tilde{\Wd}(v,s),s)\in L_v$, yet
$\tilde{\Wd}(v,s)>\Wdp(v,s)$ (by the previous property
$\tilde{\Wd}(v,s)<\Wdp(v,s)$ is not possible). Choose $i$ such that
$b(i)\Hd_i(v,s)=\Wdp(v,s)$ and $\Hd_i(v,s)\leq h'$. We have that
$(\Hd_i(v,s),s)\notin L_{v,i_{v,s}}$, as otherwise we had $\tilde{\Wd}(v,s)\leq
b(i)\Hd_i(v,s)=\Wdp(v,s)$. It follows that $|L_{v,i}|= \sigma$ and,
for each $(\Hd_i(v,t),t)\in L_{v,i}$, we have that
\begin{equation*}
(\tilde{\Wd}(v,t),t)\leq (b(i)\Hd_i(v,t),t)<(b(i)\Hd_i(v,s),s)=(\Wdp(v,s),s)
\leq (\tilde{\Wd}(v,s),s),
\end{equation*}
where in the final step we exploit the third property. As there are $\sigma$
distinct such sources $\sigma$, we arrive at the contradiction that
$(\tilde{\Wd}(v,s),s)\notin L_v$.\qedhere
\end{compactenum}
\end{proof}

\begin{lemma}\label{lemma:few_messages}
Consider the unweighted source detection algorith of~\cite{lenzen13}. The number
of messages a node $v$ broadcasts (i.e., sends to all neighbors) until it
announces the first up to $\sigma$ elements from $L_v^{(h)}$ is at most
$\BO(\sigma^2)$. Moreover, the solution to PDE given in
\theoremref{theorem:detection} can be implemented so that each node sends
$\tilde{\BO}(\sigma^2)$ messages.
\end{lemma}
\begin{proof}
Suppose the $i^{th}$ element of $L_v$ is $(d(v,s),s)$. By Lemma~4.2
from~\cite{lenzen13}, $v$ does not send $(d(v,s),s)$ in any round
$r>d(v,s)+i$.\footnote{There is a typographical error in the statement of the
lemma in the paper; the inequality should read $d_s+\ell_v^{(r)}(d_s,s)<r$.} As
it does send this value eventually, it must do so at the latest in round
$d(v,s)+i$. Afterwards, no further messages $(d_s,s)$ will be sent by $v$, as
only smaller distances could be announced, but the minimum value $d(v,s)$ was
already transmitted. Moreover, $v$ cannot send any message $(d_s,s)$ earlier
than round $d(v,s)+1$, as $d(v,s)$ rounds are required for $v$ to learn about
the existence of $s$. Hence, the total number of messages concerning the first
(up to) $\sigma$ elements of $L_v^{(h)}$ during the course of the algorithm are
at most $\sum_{i=1}^{\sigma} i \in \Theta(\sigma^2)$. Denoting by $h_v$ the
distance of the $\sigma^{th}$ element of $L_v^{(h)}$ (or $h$, if
$|L_v^{(h)}|<\sigma$), all these messages have been sent by the end of round
$h_v+\sigma$. Node $v$ does not learn about sources in distance $h_v$ or larger
earlier than round $h_v$, hence it will have sent at most $\Theta(\sigma^2)$
messages by the end of round $h_v+\sigma$. As distances to other than the
$\sigma$ closest sources are also irrelevant to other nodes, the algorithm will
still yield correct results if $v$ stops sending any messages after this number
of messages have been sent.

Using this slightly modified algorithm in the construction from
\theoremref{theorem:detection}, the second statement of the lemma follows.
\end{proof}

\begin{corollary}\label{coro:detection}
For any $0<\varepsilon\in \BO(1)$, $(1+\varepsilon)$-approximate
$(S,h,\sigma)$-estimation can be solved in $\BO((h+\sigma)/\varepsilon^2
\cdot \log n+D)$ rounds. Tables of size $\BO(\sigma \log n)$ for routing with
stretch $1+\varepsilon$ from each $v\in V$ to the (up to) $\sigma$ detected
nodes can be constructed in the same time. Each node broadcasts (i.e., sends
the same message to all neighbors) in $\BO(\sigma^2/\varepsilon\cdot \log n)$
rounds during the course of the algorithm.
\end{corollary}
\begin{proof}
We apply \theoremref{theorem:detection} to the algorithm from~\cite{lenzen13},
which has running time $R(h',\sigma)=h'+\sigma \in \BO(h/\varepsilon +\sigma)$.
Since $\varepsilon\in \BO(1)$, we have that $\BO(\log_{1+\varepsilon}n)=\BO(\log
n / \varepsilon)$; this shows the first part of the claim. For the second,
observe that if all nodes store their lists $L_{v,i}$ and send them to their
neighbors, it is trivial to derive the respective routing tables. The third
statement readily follows from \lemmaref{lemma:few_messages}.
\end{proof}

\section{Applications}
\label{sec:apps}
In this section we apply \corollaryref{coro:detection} to various
aspects of routing in weighted graphs. We improve on the best known
results for three questions: the running time required to compute
small-stretch routes with and without node relabeling, and the stretch
we can achieve within  a given running time bound and routing table
size.

\subsection{Almost Exact APSP: Routing Without Node Relabeling}
First we state our results for distributed computation of all-pairs
$(1+\eps)$-approximate shortest paths. The result follows simply by
instantiating \corollaryref{coro:detection} with all nodes as
sources and $h=\sigma=n$. As $h_{v,w}<n$ for all $v,w\in V$, $\Wdp(v,w)\leq
(1+\varepsilon)\Wd(v,w)<\infty$. The returned lists thus contain entries for all
$n=\sigma$ nodes.
\begin{theorem}\label{thm-wasps}
$(1+\varepsilon)$-approximate APSP can be solved deterministically in
$\BO(n/\varepsilon^2 \cdot \log n)$ rounds.
\end{theorem}
We note that \theoremref{thm-wasps} improves on the best known result for
computing approximate shortest paths in the \Congest model \cite{nanongkai14} in
two ways: first, it is deterministic, and second, the running time is reduced by
a logarithmic factor.

\subsection{Routing Table Computation With Node Relabeling}
In this section we use \corollaryref{coro:detection} to improve upon
the best known previous result to compute routing
tables when node relabeling is allowed~\cite{lenzen12}. We comment
that node relabeling is a common practice in routing schemes: the idea
is to encode some location information in the name of the nodes so as
to reduce the space consumed by routing tables. For example, the
Internet's IP addresses consist of a ``network'' part and a ``host''
part, which allows for hierarchical routing. In this case
the label length is another performance measure that should be noted.

In~\cite{lenzen12}, the algorithm guarantees the following. Given an integer
$0<k\le\log n$, the algorithm computes (w.h.p.) in $\tilde\BO(n^{1/2
\cdot(1+1/k)}+D)$ rounds node labels of size $\BO(\log n\log k)$ and routes with
stretch $\BO(k\log k)$. We will show how to execute this task without changing
the running time, but improve the stretch and node label size  by a $\log k$
factor.

First, let us briefly review the approach from~\cite{lenzen12}.
\begin{compactenum}
\item Sample $\tilde{\Theta}(\sqrt{n})$ nodes uniformly, forming the
\emph{skeleton} $S$.
\item \label{st-long}
Construct and make known to all nodes an $\alpha$-spanner of the
\emph{skeleton graph} $(S,E_S,W_S)$. Here, $\{s,t\}\in E_S$ if $\Hd(s,t)\leq
h\in \tilde{\Theta}(\sqrt{n})$ and $W_S(s,t)$ is the minimum weight of an
$s$-$t$ path of at most $h$ hops. It is shown that, w.h.p., distances in the
skeleton graph are identical to distances in the original graph.
\item \label{st-short}
For each $v\in V$, denote by $s_v$ the node in $S$ minimizing
$(\Wd(v,s_v),s_v)$. For each node $v$, compute distances and
routing tables  with stretch $\beta$ to all nodes $w\in V$ with
$(\Wd(v,w),w)\leq (\Wd(v,s_v),s_v)$ and from $s_v$ to $v$ (using ``tree
routing''). This part is called the \emph{short range} part of the scheme. 
\item It is then shown that if $(\Wd(v,w),w)>(\Wd(v,s_v),s_v)$, routing on
on a shortest paths from $v$ to $s_v$, $s_v$ to $s_w$ in the skeleton spanner,
(whose edges map to paths of the same weight in $G$), and finally from $s_w$ to
$w$ has stretch $\BO(\alpha \beta)$ (the same holds for distance
estimation). This part is called \emph{long range routing}.
\end{compactenum}
To facilitate routing, the label of each node $v$ contains the following
components: the identity of the closest skeleton node\footnote{For convenience,
we assume that always $S\neq \emptyset$, which holds w.h.p.} $s_v$ and the
distance to it $\Wd(v,s_v)$; and a label for the short-range routing (required for
the tree routing).

The spanner construction required for Step \ref{st-long} can be used as black
box, giving stretch $\alpha\in \Theta(1/k)$ within $\tilde{\BO}(n^{1/2+1/k}+D)$
time. Moreover, it is known how to construct labels for tree routing of size
$(1+o(1))\log n$ in time $\tilde{\BO(h)}$ in trees of depth $h$~\cite{thorup01}.

We follow the general structure of~\cite{lenzen12} by implementing Step
\ref{st-short} so that $\beta \in \BO(1)$ and the shortest-path trees do not
become to deep. To this end, we apply \corollaryref{coro:detection} with
$h=\sigma\approx \sqrt{n}$ and source set $V$. Note that the
approximation error is not limited to giving approximate distances to the
$\sigma$ closest nodes: we may obtain distance estimates to an \emph{entirely
different} set of nodes. However, we can show that the distances of the nodes
showing up in the list are at most factor $1+\varepsilon$ larger than their true
distances. Denoting by $s_v'\in S$ denote the node minimizing $(\Wdp(v,s),s)$
among nodes in $S$, the crucial properties we use are captured in the following
lemma.
\begin{lemma}\label{lemma:sampling_ok}
Suppose we sample each node into $S$ with independent probability $p$ and solve
$(1+\varepsilon)$-approximate $(V,h,\sigma)$-estimation with
$\min\{h,\sigma\}\geq c \log n / p$, where $c$ is a sufficiently large
constant. Then for all $v,w\in V$, w.h.p.\ the following statements hold.
\begin{compactenum}
\item $(\Wdp(v,w),w)\leq (\Wdp(v,s_v'),s_v')\Rightarrow \Wdp(v,w)\leq
(1+\varepsilon)\Wd(v,w)\wedge (\Wd(v,w),w)\in L_v$
\item $(\Wd(v,w),w)\leq (\Wd(v,s_v),s_v)\Rightarrow \Wdp(v,w)\leq
(1+\varepsilon)\Wd(v,w)$
\item $(\Wd(v,w),w)> (\Wd(v,s_v),s_v)\Rightarrow \Wdp(v,s_v)\leq
(1+\varepsilon)\Wd(v,w)$
\item $(\Wdp(v,w),w)> (\Wdp(v,s_v'),s_v')\Rightarrow \Wdp(v,s_v')\leq
(1+\varepsilon)\Wd(v,w)$
\end{compactenum}
\end{lemma}
\begin{proof}
Fix $v\in V$ and order $\{(\Wdp(v,w),w)\,|\,w\in V\}$ in ascending lexicographic
order. Suppose $s_v'\in S$ is the $i^{th}$ element of the resulting list. Note
that, by definition, $s_v'$ minimizes $(\Wdp(v,s),s)$ among nodes $s\in S$.
Hence, the probability that $i\geq \min\{h,\sigma\}$ is
\begin{equation*}
(1-p)^{\min\{h,\sigma\}}\in e^{-\Theta(c\log n)} = n^{-\Theta(c)}.
\end{equation*}
Since $c$ is a sufficiently large constant, this implies that $i<h$ w.h.p., and
thus $h_{v,w}<i<\min\{h,\sigma\}$ for all $w$ with $(\Wdp(v,w),w)\leq
(\Wdp(v,s_v'),s_v')$. By the properties of $(V,h,\sigma)$-estimation, it follows
that, w.h.p., $\Wdp(v,w)\leq (1+\varepsilon)\Wd(v,w)$ and $(\Wdp(v,w),w)\in
L_v$ for all such $w$.

To show the second statement, we perform the same calculation for the list
$\{(\Wd(v,w),w)\,|\,w\in V\}$; the element from $S$ minimizing $(\Wd(v,s),s)$ is
$s_v$. For the third statement, we apply the second to $s_v$, yielding that
\begin{align*}
\Wdp(v,s_v)\leq (1+\varepsilon)\Wd(v,s_v)\leq (1+\varepsilon)\Wd(v,w)
\end{align*}
w.h.p. For the final statement, if $\Wdp(v,w)\leq (1+\varepsilon)\Wd(v,w)$, it
follows that
\begin{equation*}
\Wdp(v,s_v')\leq \Wdp(v,w)\leq (1+\varepsilon)\Wd(v,w).
\end{equation*}
Otherwise, the second statement shows that $\Wd(v,w)\geq \Wd(v,s_v)$ w.h.p.,
implying
\begin{equation*}
\Wdp(v,s_v')\leq \Wdp(v,s_v)\leq (1+\varepsilon)\Wd(v,s_v)\leq
(1+\varepsilon)\Wd(v,w).\qedhere
\end{equation*}
\end{proof}

Our goal now is to let $s_v'$ take the place of $s_v$ in the original scheme. By
the previous lemma and the time bound of $\tilde{\BO}((h+\sigma)/\varepsilon^2)$
for $(1+\varepsilon)$-appromate $(S,h,\sigma)$-estimation, this achieves
$\beta\in 1+o(1)$ within the desired time bound for $p\approx 1/\sqrt{n}$.
However, this comes with a twist: as $s_v'$ may be different from $s_v$, we must
show that the resulting approximation ratio is still $\BO(\alpha)$, and as we do
not use exact distances, the $|S|$ trees induced by the approximately shortest
paths from each $v$ to $s_v'$ might overlap. The following two lemmas address
these issues, as well as the depth of the trees.
\begin{lemma}\label{lemma:stretch}
Suppose we sample each node into $S$ with independent probability $p$ and solve
$(1+\varepsilon)$-approximate $(S,h,\sigma)$-estimation with $h=c \log n / p$,
where $c$ is a sufficiently large constant. Denote by $\Wdp_S$ the
respective distance function, and by $\Wdp$ and $L_v$ the distance function and
output of $v\in V$, respectively, of a solution to $(1+\varepsilon)$-approximate
$(V,h,h)$-estimation. If for $v,w\in V$ it holds that $(\Wdp(v,w),w)\notin L_v$,
then w.h.p.\ there exist $s_0,\ldots,s_{j_0}=s_w'\in S$ such that
\begin{compactitem}
\item $\Wdp(w,s_w')\in (2+\BO(\varepsilon))\Wd_S(v,w)$ and
\item $\Wdp_S(v,s_0)+\sum_{j=1}^{j_0} \Wdp_S(s_{j-1},s_j)\in
(3+\BO(\varepsilon))\Wd(v,w)$.
\end{compactitem}
\end{lemma}
\begin{proof}
By Statements 1 and 4 of \lemmaref{lemma:sampling_ok}, $(\Wdp(v,w),w)\notin
L_v$ means that $\Wdp(v,s_v')\leq (1+\varepsilon)\Wd(v,w)$ w.h.p. Hence,
\begin{equation*}
\Wd(v,s_v)\leq \Wd(v,s_v')\leq \Wdp(v,s_v')\leq (1+\varepsilon)\Wd(v,w).
\end{equation*}
By the triangle inequality, it follows that
\begin{equation*}
\Wd(w,s_w)\leq \Wd(w,s_v)\leq \Wd(v,w)+\Wd(v,s_v)\leq (2+\varepsilon)\Wd(v,w).
\end{equation*}
w.h.p. Applying the second statement of \lemmaref{lemma:sampling_ok} to $w$ and
$s_w$, we obtain that
\begin{equation*}
\Wd(w,s_w')\leq \Wdp(w,s_w')\leq \Wdp(w,s_w)\leq
(1+\varepsilon)\Wd(w,s_w)\leq 2(1+\varepsilon)^2\Wd(v,w),
\end{equation*}
and, from the first statement,
\begin{equation*}
\Wdp(w,s_w')\leq (1+\varepsilon)\Wd(w,s_w')\leq
2(1+\varepsilon)^3\Wd(v,w),
\end{equation*}
i.e., the first part of the lemma's claim holds (recall that $\varepsilon\in
\BO(1)$). Moreover, it follows that
\begin{equation*}
\Wd(v,s_w')\leq \Wd(v,w)+\Wd(w,s_w')\leq 3(1+\varepsilon)^2\Wd(v,w).
\end{equation*}
Consider a shortest path from $v$ to $s_w'$, and denote by
$s_0,\ldots,s_{j_0}\in S$ the sampled nodes that are encountered when traversing
it from $v$ to $s_w'$; in particular, $s_{j_0}=s_w'$. By the same calculation as
for \lemmaref{lemma:sampling_ok}, w.h.p.\ any two consecutive sampled nodes are
no more than $h$ hops apart. As the path is a shortest path from $v$ to $s_w'$,
the subpaths from $s_{j-1}$ to $s_j$, $j\in \{1,\ldots,j_0\}$, and from $v$ to
$s_0$ are also shortest paths. Therefore, $h_{v,s_0}\leq h$ and, for each $j$,
$h_{s_{j-1},s_j}\leq h$. We conclude that
\begin{align*}
\Wdp_S(v,s_0)+\sum_{j=1}^{j_0} \Wdp_S(s_{j-1},s_j)
&\leq (1+\varepsilon)\left(\Wd(v,s_0)+ \sum_{j=1}^{j_0}
\Wd(s_{j-1},s_j)\right)\\
&=(1+\varepsilon)\Wd(v,s_v')\\
&\leq 3(1+\varepsilon)^3\Wd(v,w),
\end{align*}
i.e., the second part of the claim of the lemma holds.
\end{proof}

\begin{lemma}\label{lemma:trees}
Suppose we sample each node into $S$ with independent probability $p$ and solve
$(1+\varepsilon)$-approximate $(V,h,\sigma)$-estimation with
$h = c \log n / p$, where $c$ is a sufficiently large
constant. For $s\in S$, denote by $T_s$ the tree induced by the routing paths
from $v$ to $s$ for all $v\in V$ with $s_v'=s$. The depth of $T_s$ is bounded by
$\BO(h\log n/\varepsilon)$, and each node participates in at most
$\BO(\log n)$ different trees.
\end{lemma}
\begin{proof}
Recall that routing from $v$ to $s_v'$ is based on the routing tables
$L_{v,i}$ determined by the unweighted source detection instances on $G_i$,
$i\in \{0,\ldots,i_{\max}\}$. The induced shortest-path trees in $G_i$ have
depth at most $h'\in \BO(h/\varepsilon)$, and they cannot overlap. By
construction, the respective paths in $G$ cannot have more hops. However, it is
possible that when routing from $v$ to $s_v'$, some node on the way knows
of a shorter path to $s_v'$ due to a source detection instance on $G_j$,
$j\neq i$, and therefore ``switches'' to the shortest-path tree in $G_j$.
Because $\Wd_j(v,w)\geq \Wd_i(v,w)$ for all $v,w\in V$ and $j\geq i$, we may
however w.l.o.g.\ assume that the index $i$ such that routing decisions are made
according to $L_{v,i}$ is decreasing on each routing path from some node $v$ to
$s_v'$. Thus, the total hop count of the path is bounded by $\BO(i_{\max}
h')\subseteq \BO(h\log n/\varepsilon)$. Consequently, the depth of each $T_s$ is
bounded by this value.

Concerning the number of trees a node may participate in, observe that if some
node $v$ decides that the next routing hop to $s_v'$ is its neighbor $u$,
it does so because $s_v'$ minimizes the hop distance from $v$ to $s_v'$ in
$G_i$, according to its list $L_{v,i}$. As there are $i_{\max}+1\in \BO(\log n)$
different lists $L_{v,i}$, this is also a bound on the number of different trees
$v$ may participate in.
\end{proof}

We summarize with the following theorem.
\begin{theorem}
For any $k\in \N$, routing table construction with stretch $6k-1+o(1)$
and labels of size $\BO(\log n)$ can be solved in
$\tilde{\BO}(n^{1/2+1/(4k)}+D)$ rounds.
\end{theorem}
\begin{proof}
We follow the approach outlined above, sampling nodes into $S$ with probability
$p = n^{-1/2-1/(4k)}$. By Chernoff's bound, this implies that $|S|\in \Theta
(n^{1/2-1/(4k)})$ w.h.p. Using \corollaryref{coro:detection}, we solve
$(1+\varepsilon)$-approximate $(V,h,\sigma)$-estimation with $h=\sigma = c \log
n /p$, for $c\in \BO(1)$ sufficiently large, and, say, $\varepsilon = 1/\log n$.
This takes $\tilde{\BO}(1/p)=\tilde{\BO}(n^{-1/2-1/(4k)}+D)$ rounds and enables
for each $v\in V$ to route to all nodes $w\in V$ with $(\Wdp(v,w),w)\in L_v$
along a path of weight at most $\Wdp(v,w)$. By \lemmaref{lemma:sampling_ok},
this enables for each $v,w\in V$ with $(\Wdp(v,w),w)\leq (\Wdp(v,s_v'),s_v')$ to
determine that this condition is satisfied and route from $v$ to $w$ with
stretch $(1+\varepsilon)$.

To handle the possibility that $(\Wdp(v,w),w)\notin L_v$, we call upon
\corollaryref{coro:detection} once more. This time we solve
$(1+\varepsilon)$-approximate $(S,h,|S|)$-detection; denote by $\Wdp_S$ the
corresponding distance function. Since $|S|\in \BO(n^{1/2-1/(4k)})$ w.h.p., this
requires $\tilde{\BO}(n^{-1/2-1/(4k)})$ rounds w.h.p. We apply
\lemmaref{lemma:stretch}, showing that there are $s_0,\ldots,s_{j_0}=s_w'\in S$
so that $\Wdp(s_w',w)\in (2+\BO(\varepsilon))\Wd(v,w)$ and
$\Wdp_S(v,s_0)+\sum_{j=1}^{j_0} \Wdp_S(s_{j-1},s_j)\in
(3+\BO(\varepsilon))\Wd(v,w)$. If we can route from $v$ to $s_w'$ incurring an
additional stretch factor of $2k-1$ and from $s_w'$ to $w$ over a path of weight
$\Wdp(s_w',w)$, the total stretch will be
\begin{equation*}
(2+\BO(\varepsilon))+(2k-1)(3+\BO(\varepsilon))\in
6k-1+\BO(\varepsilon)\subset 6k-1+o(1),
\end{equation*}
i.e., the routing scheme satisfies the claimed stretch bound.

Concerning routing from $s_w'$ to $w$, we construct labels for tree routing
using the algorithm from~\cite{thorup01} that terminates in $\tilde{\BO}(h)$
rounds in trees of depth $h$. By \lemmaref{lemma:trees}, this can be
done in $\tilde{\BO}(h)=\tilde{\BO}(n^{-1/2-1/(4k)})$ rounds, where we simulate
one round on each of the trees using $\BO(\log n)$ rounds, one for each tree
single node may participate in. The computed label of size $(1+o(1))\log n$ is
added to the label of $w$, permitting to route from $s_w'$ to $w$ over a path of
weight at most $\Wdp(w,s_w')$.

To route from $v$ to $s_w'$, consider the graph on node set $S$ with edge set
$\{\{s,t\}\,|\,\Wdp_S(s,t)<\infty\}$, where the edge weights are given by
$\Wdp_S$. For this graph, each node $s\in S$ knows its incident edges and their
weights. Using the simulation of the Baswana-Sen algorithm~\cite{baswana07}
given in~\cite{lenzen12}, we can construct and make known to all nodes a $2k-1$
spanner\footnote{I.e., a subgraph in which distances increase by at most a
factor $2k-1$.} of this graph in
\begin{equation*}
\tilde{\BO}\left(|S|^{1+1/k}+D\right)
=\tilde{\BO}\left(n^{(1/2-1/(4k))(1+1/k)}+D\right)
\subset \tilde{\BO}\left(n^{1/2+1/(4k)}+D\right)
\end{equation*}
rounds. Using this knowledge, the fact that $v$ knows $\Wdp_S(v,s_0)$, and the
routing tables from the second application of \corollaryref{coro:detection},
w.h.p.\ we can route with the desired stretch from $v$ to $s_w'$ based on the
identifier of $s_w'$, which we add to the label of $v$. This completes the proof of the
stretch bound. Checking the individual bounds we picked up along the way, we see
that the label size is $\BO(\log n)$ and the running time is
$\tilde{\BO}(n^{1/2+1/(4k)}+D)$ w.h.p.
\end{proof}

\subsection{Compact Routing on Graphs of Small Diameter}
We now turn to the question of how to minimize the routing table size
when computing routing tables distributedly. The idea in the
algorithm is to construct an (approximate) Thorup-Zwick routing
hierarchy~\cite{thorup01}. We remark that compared to the original construction,
we lose a factor of roughly $2$ in stretch. This comes from the fact that in the
centralized setting, one assumes access to the table of both nodes, i.e., the
origin and destination of the routing or distance query. This is equivalent to
identifying tables and lables, resulting in lable size
$\tilde{\Theta}(n^{1/k})$. We consider this inexpedient for distributed systems
and hence focus on obtaining small lables.

Our approach is efficient if $D$ is small. Using exact distances, the
construction would look as follows.
\begin{compactenum}
\item For each node $v\in V$, choose its \emph{level} by an
independent geometric distribution, i.e., the probability to have
level at least $l\in \{0,\ldots,k-1\}$ equals $p_l:=n^{-l/k}$. Denote the set of
nodes of level at least $l$ by $S_l$; in particular, $S_0=V$.
\item For each node $v$ and each level $l\in \{1,\ldots,k-1\}$, determine the
closest node $s_l(v)$ to $v$ and the set $S_{l-1}(v)$ of all nodes in $S_{l-1}$
closer to $v$ than $s_l(v)$ (ties broken by node identifiers); for
notational convenience, let $s_0(v):=v$ and $S_{k-1}(v):=S_{k-1}$.
\item Determine tables and labels for routing and
distance approximation (i) from $v$ to all nodes in $S_l(v)$, for all $l\in
\{0,\ldots,k-1\}$, and (ii) from $s_l(v)$ to $v$, where $l\in \{1,\ldots,k-1\}$.
To determine the final label of $v$, concatentate its individual labels and the
labels for routing from $s_l(v)$, $l\in \{1,\ldots,k-1\}$.
\end{compactenum}
In our implementation, we replace exact distances by
$(1+\varepsilon)$-approximate distances for sufficiently small
$\varepsilon$. Henceforth, we assume that the sets $S_l'(v)$ and nodes $s_l'(v)$
are defined as above, but with respect to $\Wdp_l$, the distance function
corresponding to the instance of partial distance estimation we solve for level $l\in
\{0,\ldots,k-1\}$. Let us first examine the effect of the inaccurate distances
on the stretch. This is done by a repeated application of the argument of
\lemmaref{lemma:stretch}.
\begin{lemma}\label{lemma:stretch_iterative}
For $l\in \{1,\ldots,k-1\}$, denote by $\Wdp_l$ the distance function
corresponding to a $(1+\varepsilon)$-approximate solution to
$(S_l,h_l,\sigma_l)$-estimation, where $h_l=\sigma_l=c\log n /p_l$ for a
sufficiently large constant $c$. Suppose $v,w\in V$ and $\ell\in
\{0,\ldots,k-1\}$ is minimal so that $s_{\ell}'(w)\in S_{\ell}'(v)$. Then,
w.h.p.,
\begin{equation*}
\Wd(v,s_{\ell}'(w))+\Wd(s_{\ell}'(w),w)\leq
(1+\varepsilon)^{4\ell}(4\ell+1)\Wd(v,w).
\end{equation*}
\end{lemma}
\begin{proof}
For $0\leq l\leq \ell$, we prove by induction on $l$ that
\begin{equation*}
\Wd(w,s_l'(w))\leq (1+\varepsilon)^{2l}2l\Wd(v,w)
\end{equation*}
and that
\begin{equation*}
\Wd(v,s_l'(w))\leq (1+\varepsilon)^{2l}(2l+1)\Wd(v,w)
\end{equation*}
w.h.p. For $l=0$, trivially $\Wd(v,s_0'(w))=\Wd(v,w)$ and
$\Wd(w,s_0'(w))=\Wd(w,w)=0$. For the step from $l$ to $l+1$, we make the
intermediate claim that
\begin{equation*}
\Wd(v,s_{l+1}'(v))\leq (1+\varepsilon)^{2l+1}(2l+1)\Wd(v,w).
\end{equation*}
If $\Wdp_{l+1}(v,s_l'(w))\leq (1+\varepsilon)\Wd(v,s_l'(w))$, then
\begin{equation*}
\Wd(v,s_{l+1}'(v))\leq
\Wdp_{l+1}(v,s_{l+1}'(v))\leq \Wdp_{l+1}(v,s_l'(w))\leq
(1+\varepsilon)\Wd(v,s_l'(w))
\end{equation*}
w.h.p., where the second last step exploits that $s_l'(w)\notin S_l'(v)$ by the
definition of $\ell>l$, but $s_{l+1}'(v)\in S_l'(v)$ w.h.p. Otherwise,
Statements 1 and 4 of \lemmaref{lemma:sampling_ok} yield that
$\Wdp_{l+1}(v,s_{l+1}'(v))\leq (1+\varepsilon)\Wd(v,s_l'(w))$, resulting in the
same bound on $\Wd(v,s_{l+1}'(v))$. Either way, applying the induction
hypothesis shows the claim.

We now can apply the triangle inequality to see that
\begin{equation*}
\Wd(w,s_{l+1}'(v))\leq \Wd(v,w)+\Wd(v,s_{l+1}'(v))\leq
(1+\varepsilon)^{2l+1}(2(l+1))\Wd(v,w).
\end{equation*}
If $\Wdp_{l+1}(w,s_{l+1}'(v)) \leq (1+\varepsilon)\Wd(w,s_{l+1}'(v))$, then
\begin{equation*}
\Wd(w,s_{l+1}'(w))\leq \Wdp_{l+1}(w,s_{l+1}'(w))\leq \Wdp_{l+1}(w,s_{l+1}'(v))
\leq (1+\varepsilon)\Wd(w,s_{l+1}'(v)).
\end{equation*}
Otherwise, Statements 1 and 4 of \lemmaref{lemma:sampling_ok} imply that
$\Wdp_{l+1}(w,s_{l+1}'(w))\leq (1+\varepsilon)\Wd(w,s_{l+1}'(v))$ w.h.p.,
resulting in the same bound on $\Wd(w,s_{l+1}'(w))$. Therefore, in both cases,
\begin{equation*}
\Wd(w,s_{l+1}'(w))\leq (1+\varepsilon)^{2(l+1)}2(l+1)\Wd(v,w),
\end{equation*}
i.e., the first part of the hypothesis is shown for index $l+1$. The second part
readily follows by applying the triangle inequality once more, yielding
\begin{equation*}
\Wd(v,s_{l+1}'(w))\leq \Wd(v,w)+\Wd(w,s_{l+1}'(w))\leq
(1+\varepsilon)^{2(l+1)}(2(l+1)+1)\Wd(v,w)
\end{equation*}
w.h.p. This concludes the induction. Evaluating both statements of the induction
hypothesis for index $l=\ell$ completes the proof.
\end{proof}
\lemmaref{lemma:stretch_iterative} shows that for $\varepsilon \in o(1/k)$,
routing from $v$ to $w$ via $s_{\ell}'(w)\in S_{\ell}'(v)$ for minimal $\ell$
achieves stretch $4k-3+o(1)$.
It remains to construct the hierarchy efficiently. We start with a
general algorithm.
\begin{lemma}\label{lemma:low_levels}
  For each level $l\in \{0,\ldots,k-1\}$, we can determine w.h.p.\ for
  all nodes $v$ the set $S_l'(v)$ and the respective distance and
  routing information in $\tilde{\BO}(\varepsilon^{-2}n^{(l+1)/k})$
  rounds, where the tables have size $\BO(n^{1/k}\log^2 n)$. Within
  this time, we can also determine labels of size $(1+o(1))\log n$ and
  tables of size $\BO(\log^2 n)$ at each node for routing from
  $s_l'(v)$ to $v$.
\end{lemma}
\begin{proof}
For a sufficiently large constant $c$, we perform $(1+\varepsilon)$-approximate
$(S_l,h_{l+1},\sigma)$-estimation with $h_{l+1}=cn^{(l+1)/k}\log n$ and
$\sigma=cn^{1/k}\log n$. For $l<k-1$, the probability that
$(\Wdp(v,s_{l+1}'(v)),s_{l+1}'(v))$ has index $i\geq \sigma$ if we order
$\{(\Wdp(v,s),s)\,|\,s\in S_l\}$ ascendingly is $(1-p_l/p_{l+1})^{\sigma} \in
n^{-\Omega(c)}$. The probability that $(\Wdp(v,s_{l+1}'(v)),s_{l+1}'(v))$ has
index $j\geq h_{l+1}$ if we order $\{(\Wdp(v,w),w)\,|\,w\in V\}$ ascendingly is
$(1-1/p_{l+1})^{h_{l+1}}\in n^{-\Omega(c)}$. By appending a bit to messages
indicating whether $s\in S_l$ is also in $S_{l+1}$, we can thus use
\corollaryref{coro:detection} to show that, w.h.p., we obtain suitable tables
for routing from $v\in V$ to $S_l(v)$ and $s_{l+1}(v)$ within the stated time
bound. If $l=k-1$, we have that $h_{l+1}>n$ and $|S_l|=|S_{k-1}|\leq \sigma$
w.h.p.; in this case, \corollaryref{coro:detection} shows that the construction
can be performed as well.

Regarding the second part of the statement, observe that analogously to
\lemmaref{lemma:trees}, the routing trees rooted at each node $s_{l+1}\in
S_{l+1}$ have depth $\BO(h_{l+1}/\varepsilon)$ and each node participates in at
most $\BO(\log n)$ of them. Thus, we can apply the construction from~\cite{thorup01}
to obtain labels (and tables) of size $(1+o(1))\log n$ for tree routing on each
of the trees in $\tilde{\BO}(h_{l+1}/\varepsilon)\subseteq
\tilde{\BO}(\varepsilon^{-2}n^{(l+1)/k})$ rounds. As each node participates in
$\BO(\log n)$ trees, the table size for this routing information is
$\BO(\log^2 n)$.
\end{proof}

\begin{theorem}\label{thm-spd}
Tables of size $\tilde{\BO}(n^{1/k})$ and labels of size $\BO(k\log n)$ for
routing with stretch $4k-3$ can be computed in the \Congest model in
$\tilde{\BO}(\SPD+n^{1/k})$ rounds.
\end{theorem}
\begin{proof}
We choose $\varepsilon \in \Theta(1/\log n)$. Then
\lemmaref{lemma:stretch_iterative} shows that the stretch of the routing scheme
will be $2(k-1)+2(k-1)+1=4k-3$. To construct the labels and tables, we use the
approach of \lemmaref{lemma:low_levels}, but with $h:=\SPD$. This is feasible,
as by definition $h\geq h_{v,w}$ for all $v,w\in V$. As $\sigma \in
\tilde{\Theta}(n^{1/k})$, and since we may assume w.l.o.g.\ that $k\in
\BO(\log n)$, the claimed running time bound follows.
\end{proof}

Unfortunately, \theoremref{thm-spd} gives a good running time guarantee only
when $\SPD$ is small. Worse, the strategy can be applied only if an upper bound
on $\SPD$ is known (and the running time depends on that bound), unlike the
algorithm of running time $\tilde{\BO}(\SPD\cdot n^{1/k})$
from~\cite{dassarma12}.\footnote{Their algorithm only handles distance queries
and assumes that also the table of the destination can be accessed (i.e., the
lables are identical to the tables). Both assumptions can be removed to achieve
the same properties as our solution within $\tilde{\BO}(\SPD\cdot n^{1/k})$
rounds.} On the other hand, applying \lemmaref{lemma:low_levels} to all levels
(without modifying $h$) results in running time $\tilde{\BO}(n)$. In the
remainder of this section, we explain how to improve on \theoremref{thm-spd} by
``short-circuiting'' the higher levels of the hierarchy. This approach yields
better results when the hop diameter is small.

We now describe the construction. Let $l_0<k-1$ be some level to be determined
later. We will ``truncate'' the construction at level $l_0$ by constructing a
\emph{skeleton graph} as follows.
\begin{definition}[$l_0$ skeleton graph]
The \emph{skeleton graph on level $l_0$} is
$G(l_0)=(S_{l_0},E_{l_0},\Wd)$, where $\{s,t\}\in E_{l_0}$ if and only
if $h_{s,t}\leq cn^{l_0/k}\log n$ for a sufficiently large constant
$c$. We denote $ h_{l_0}:=cn^{l_0/k}\log n$.
\end{definition}
The $l_0$ skeleton graph preserves the original skeleton distances, as the
following lemma states.

\begin{lemma}\label{lemma:distance_preserved}
For any $\varepsilon>0$, $h,\sigma\in \N$, and $S\subseteq S_{l_0}$,
denote by $\Wd_{S_{l_0}}$ the distance function resulting from solving
$(1+\varepsilon)$-approximate $(S_{l_0},h_{l_0},|S_{l_0}|)$-estimation and
by $\Wd_s$ the distance function resulting from solving
$(1+\varepsilon)$-approximate $(S,ch\log n,\sigma)$-estimation on $G(l_0)$,
where $c$ is a sufficiently large constant. Then, w.h.p.,
\begin{equation*}
\Wdp(v,s):=\min_{t\in S_{l_0}}\{\Wdp_{S_{l_0}}(v,t)+\Wdp_S(t,s)\}
\end{equation*}
is a suitable distance function for $(1+\varepsilon)$-approximate
$(S,h\cdot h_{l_0},\sigma)$-estimation on $G$.
\end{lemma}
\begin{proof}
By the triangle inequality, for any $v\in V$, $t\in S_{l_0}$, and $s\in S$,
\begin{equation*}
\Wd(v,s)\leq \Wd(v,t)+\Wd(t,s)\leq \Wdp_{S_{l_0}}(v,t)+\Wdp_S(t,s).
\end{equation*}
Now suppose $h_{v,s}\leq h\cdot h_{l_0}$ for some $v\in V$ and $s\in S$. The
expected number of nodes in $S_{l_0}$ on a shortest path from $v$ to $s$ of
$h_{v,s}$ hops is $p_{l_0}h_{v,s}\in \BO(h\log n)$. By Chernoff's bound, this
number is smaller than $ch\log n$ w.h.p., as $c$ is sufficiently large. Another
application of Chernoff's bound shows that the maximum hop distance between
nodes from $S_{l_0}$ on the path is bounded by $h_{l_0}$ w.h.p.

Denoting by $t_{v,s}\in S_{l_0}$ the first sampled node on the path, the above
shows that the following properties hold w.h.p.
\begin{compactitem}
\item $\Wd(v,s)= \Wd(v,t_{v,s})+\Wd(t_{v,s},s)$,
\item $\Wd_{G(l_0)}(t_{v,s},s)=\Wd(t_{v,s},s)$, where $\Wd_{G(l_0)}$ denotes
the weighted distance in $G(l_0)$,
\item $\Wdp_{S_{l_0}}(v,t_{v,s})\leq (1+\varepsilon)\Wd(v,t_{v,s})$, and
\item $\Wdp_S(t_{v,s},s)\leq (1+\varepsilon)\Wd_{G(l_0)}(t_{v,s},s)$.
\end{compactitem}
Overall, this yields
\begin{align*}
\Wdp(v,s) &=\min_{t\in S_{l_0}}\{\Wdp_{S_{l_0}}(v,t)+\Wdp_S(t,s)\}\\
&\leq \Wdp_{S_{l_0}}(v,t_{v,s})+\Wdp_S(t_{v,s},s)\\
&\leq (1+\varepsilon)(\Wd(v,t_{v,s})+\Wd_{G(l_0)}(t_{v,s},s))\\
&= (1+\varepsilon)(\Wd(v,t_{v,s})+\Wd(t_{v,s},s))\\
&=(1+\varepsilon)\Wd(v,s).\qedhere
\end{align*}
\end{proof}

\begin{corollary}\label{coro:distance_preserved}
If in the construction from \lemmaref{lemma:distance_preserved} we replace
$G(l_0)$ by the graph $\tilde{G}(l_0)$ constructed by solving
$(1+\varepsilon)$-approximate $(S_{l_0},h_{l_0},|S_{l_0}|)$-estimation and
assigning weight $\Wdp_{S_{l_0}}(s,t)$ to edge $\{s,t\}$, the resulting function
$\Wdp$ is a suitable distance function for $(1+\varepsilon)^2$-approximate
$(S,h\cdot h_{l_0},\sigma)$-estimation.
\end{corollary}
\begin{proof}
By definition, for all edges $\{s,t\}\in E_{l_0}$, we have that
$\Wdp_{S_{l_0}}(s,t)\leq (1+\varepsilon)\Wd(s,t)$. Also, clearly
$\Wdp_{S_{l_0}}(s,t)\geq \Wd(s,t)$ for all $s,t\in S_{l_0}$. Therefore, we can
reason analogously to \lemmaref{lemma:distance_preserved}, except for incurring
another factor of $1+\varepsilon$ in stretch.
\end{proof}

Next, we consider the simulation of the truncated levels in the hierarchy.
\begin{lemma}\label{lemma:long_distance}
For any integer $l_0\geq k/2+1$, we can construct level $l\geq l_0$ of the
routing hierarchy in $\tilde{\BO}(\varepsilon^{-2}(n^{l_0/k}+ n^{(k-l_0)/k}D))$
rounds w.h.p., where the tables and labels are of size $\BO(n^{1/k})$ and
$\BO(\log n)$, respectively.
\end{lemma}
\begin{proof}
Recall that $\varepsilon\in \BO(1)$. We choose $\varepsilon'\in
\Theta(\varepsilon)$ such that $(1+\varepsilon')^2=(1+\varepsilon)$. We solve
$(1+\varepsilon')$-approximate $(S_{l_0},h_{l_0},|S_{l_0}|)$-estimation using
\corollaryref{coro:detection}, w.h.p.\ requiring
\begin{equation*}
\tilde{\BO}(\varepsilon^{-2}(h_{l_0}+|S_{l_0}|)+D)=
\tilde{\BO}\left(\varepsilon^{-2}\left(n^{l_0/k}+n^{(k-l_0)/k}\right)+D\right)
\subseteq \tilde{\BO}\left(\varepsilon^{-2}n^{l_0/k}+D\right)
\end{equation*}
rounds. Our goal is to apply \corollaryref{coro:distance_preserved}. To this
end, we will simulate, for $h=h_{l+1}/h_{l_0}$ and sufficiently a sufficiently
large constant $c$, $(1+\varepsilon')$-approximate $(S_l,c\,h\log
n,c\,n^{1/k}\log n)$-estimation on $\tilde{G}(l_0)$, in a way such that
\emph{all} nodes will learn the output of \emph{all} nodes in $S_{l_0}$. As in
\lemmaref{lemma:low_levels}, a bit indicating whether a source is in $S_{l+1}$
is added to messages if $l<k-1$.

Before we explain how to do this, let us show how this permits to construct
level $l$ of the routing hierarchy. From the collected information, w.h.p.\
nodes can locally compute the distance function $\Wdp$ from
\corollaryref{coro:distance_preserved} for the $\sigma$ closest nodes in $S_l$
w.r.t.\ $\Wdp$ and, analogously to \lemmaref{lemma:low_levels}, derive their
table for routing from $v$ to $S_l'$ and $s_l'(v)$.

To enable tree routing from $s_l'(v)$ to $v$, split the tree rooted at $s_l'(v)$
into the unique maximal subtrees rooted at $s\in S_{l_0}$ that contain no
internal nodes from $S_{l_0}$ (i.e., all such nodes are either the root or
leaves). By \lemmaref{lemma:trees}, these subtrees have depth at most
$\tilde{\BO}(h_{l_0})$. We use separate labeling schemes for the (globally
known) tree on $\tilde{G}(l_0)$ that describes the connections between nodes in
$S_{l_0}$ in the routing tree rooted at $s_l'(v)$ and the subtrees rooted at
each $s\in S_{l_0}$. The former can be computed locally. The latter can be
labeled in time $\tilde{\BO}(\varepsilon^{-2}h_{l_0})$, provided that each node
participates in $\tilde{\BO}(1)$ different trees only. Analogously to
\lemmaref{lemma:trees}, this holds true because each routing decision must
correspong to one of the $\BO(\log n)$ top entries of the routing tables (either
for routing in $G$ to some node in $S_{l_0}$ or in $\tilde{G}(l_0)$). This
approach requires each node in the tree to store two labels of size
$(1+o(1))\log n$. Routing can now be executed by determining the next node from
$S_{l_0}$ to visit on the path from $s_l'(v)$ to $v$ (if there still is one) and
then use the label for the current subtree to find the next routing hop.

It remains to discuss how to solve $(1+\varepsilon')$-approximate
$(S_l,h,c\,n^{1/k}\log n)$-estimation on $\tilde{G}(l_0)$ quickly. Recall that
each node in $S_{l_0}$ knows its neighbors and the weights of incident edges
from the solution of $(1+\varepsilon')$-approximate
$(S_{l_0},h_{l_0},|S_{l_0}|)$-estimation computed earlier. We simulate the
algorithm given by \corollaryref{coro:detection}, exploiting that each node
broadcasts in only $\tilde{\BO}(n^{2/k})$ rounds in total. For each simulated
round $i\in \{1,\ldots,h+\sigma\}$, we pipeline the communication over a BFS
tree, which takes $\BO(M_i+D)$ rounds in $G$, where $M_i$ is the number of nodes
in $\tilde{G}(l_0)$ that broadcast in simulated round $i$; this time bound
includes $\BO(D)$ rounds for global synchronization of when the next simulated
round starts. Therefore, the total number of communication rounds in $G$ is
\begin{equation*}
\sum_{i=1}^{h+\sigma} \BO(M_i+D) \subseteq
\tilde{\BO}(\sigma^2|S_{l_0}|+(h+\sigma)D)
\subseteq \tilde{\BO}(n^{2/k}\cdot n^{(k-l_0)/k}+n^{(l-l_0+1)/k}D))
\subseteq \tilde{\BO}(n^{l_0/k}+n^{(k-l_0)/k}D))
\end{equation*}
w.h.p., as $|S_{l_0}|\in \tilde{\BO}(n^{(k-l_0)/k})$ w.h.p. The bounds on table
and lable size follow from \lemmaref{lemma:low_levels} and the above discussion
of the tree labeling scheme.
\end{proof}

We can now put all the pieces together to obtain the following result.
\begin{theorem}\label{thm-compact}
Suppose we are given $k\in \N$ and some integer $k/2+1\leq l_0\leq k$. Then
tables of size $\tilde{\BO}(n^{1/k})$ and labels of size $\BO(k\log n)$ enabling
routing and distance approximation with stretch $4k-3+o(1)$ can be constructed in
$\tilde{\BO}(n^{l_0/k}+ n^{(k-l_0)/k}D)$ rounds w.h.p.
\end{theorem}
\begin{proof}
Fix $\varepsilon:= 1/\log^2 n$. W.l.o.g., $k\in \BO(\log n)$. We construct the
first $l_0$ levels of the hierarchy using \lemmaref{lemma:low_levels}, and the
remaining levels using \lemmaref{lemma:long_distance}. This yields the stated
running time, by \lemmaref{lemma:stretch_iterative} and
\corollaryref{coro:distance_preserved} stretch
$(1+\varepsilon)^{\BO(k)}(4(k-1)+1)\in 4k-3+o(1)$, table size
$\tilde{\BO}(kn^{1/k})=\tilde{\BO}(n^{1/k})$, and labels of size $\BO(k\log n)$.
\end{proof}

We can now choose an appropriate value for $l_0$ depending on $D$. If the
running time is worse than about $n^{2/3}$, we handle the higher levels simply
by making $\tilde{G}(l_0)$ known to all nodes and solving locally.
\begin{corollary}
For any integer $k>1$, tables of size $\tilde{\BO}(n^{1/k})$ and labels of size
$\BO(k\log n)$ enabling routing and distance approximation with stretch
$4k-3+o(1)$ can be constructed in 
\begin{equation*}
\tilde{\BO}\left(\min\{(Dn)^{1/2}\cdot n^{1/k},n^{2/3+2/(3k)}\}+D\right)
\end{equation*}
rounds w.h.p.
\end{corollary}
\begin{proof}
If $k=2$, the minimum is attained for the second term. If $k\geq 3$, choose
$l_0$ as the closest integer to $k(\log D/\log n+1)/2$, however, at least
$k/2+1$ and at most $k-1$. By \theoremref{thm-compact}, we can compute compact
routing tables within
\begin{equation*}
\tilde{\BO}\left(n^{l_0/k}+ n^{(k-l_0)/k}D\right)\subseteq
\tilde{\BO}\left((Dn)^{1/2}n^{1/k}+D\cdot n^{1/k}\right),
\end{equation*}
where the first summand covers the case that $l_0\neq k-1$ and the second one
the possibility that $l_0=k-1$, implying that $l_0\geq k-3/2$ and hence
$Dn^{1/k}\geq n^{1-1/(2k)}> n^{l_0/k}$.
Note that in the latter case, we have that $Dn^{1/k}>n^{5/6}$, i.e., the term is
irrelevant for the minimum in the running time bound of the corollary unless
$k=3$.

Alternatively, we may handle levels $1,\ldots,l_0-1$ using
\lemmaref{lemma:low_levels}, determine $\tilde{G}(l_0)$, broadcast all its
edges over a BFS tree, and compute the table construction locally. This takes
\begin{equation*}
\tilde{\BO}(n^{l_0/k}+n^{2(k-l_0)}+D)
\end{equation*}
rounds w.h.p., as $\tilde{G}(l_0)$ has at most $|S_{l_0}|^2$ edges and
$|S_{l_0}|\in \BO(n^{(k-l_0)/k})$ w.h.p. For the optimal choice of $l_0$, this
results in running time
\begin{equation*}
\tilde{\BO}(n^{2/3+2/(3k)}+D).
\end{equation*}
For the special case of $k=3$, note that $l_0=2$ in fact yields running time
$\tilde{\BO}(n^{2/3}+D)\subset \tilde{\BO}((nD)^{1/2}n^{1/3})$. Hence we may
drop the additive term of $n^{1/k}D$ from the first bound when taking the
minimum.
\end{proof}
\newpage
\pagenumbering{roman}
\small
\bibliographystyle{abbrv}
\bibliography{main}

\begin{thebibliography}{10}

\bibitem{abraham08}
I.~Abraham, C.~Gavoille, D.~Malkhi, N.~Nisan, and M.~Thorup.
\newblock Compact name-independent routing with minimum stretch.
\newblock {\em ACM Trans. Algorithms}, 4(3):37:1--37:12, 2008.

\bibitem{ABLP89}
B.~Awerbuch, A.~{Bar-Noy}, N.~Linial, and D.~Peleg.
\newblock Compact distributed data structures for adaptive network routing.
\newblock In {\em Proc.\ 21st ACM Symp.\ on Theory of Computing}, pages
  230--240, May 1989.

\bibitem{baswana07}
S.~Baswana and S.~Sen.
\newblock A simple and linear time randomized algorithm for computing sparse
  spanners in weighted graphs.
\newblock {\em Random Structures and Algorithms}, 30(4):532--563, 2007.

\bibitem{Bellman}
R.~E. Bellman.
\newblock On a routing problem.
\newblock {\em Quart.\ Appl.\ Math.}, 16:87--90, 1958.

\bibitem{dassarma12}
A.~{Das Sarma}, M.~Dinitz, and G.~Pandurangan.
\newblock Efficient computation of distance sketches in distributed networks.
\newblock In {\em Proc.\ 24th ACM Symp.\ on Parallelism in Algorithms and
  Architectures}, 2012.

\bibitem{dassarma12hardness}
A.~{Das Sarma}, S.~Holzer, L.~Kor, A.~Korman, D.~Nanongkai, G.~Pandurangan,
  D.~Peleg, and R.~Wattenhofer.
\newblock {Distributed Verification and Hardness of Distributed Approximation}.
\newblock {\em SIAM Journal on Computing}, 41(5):1235--1265, 2012.

\bibitem{Ford-56}
L.~R. Ford.
\newblock Network flow theory.
\newblock Technical Report P-923, The Rand Corp., 1956.

\bibitem{FHW-12}
S.~Frischknecht, S.~Holzer, and R.~Wattenhofer.
\newblock Networks cannot compute their diameter in sublinear time.
\newblock In {\em Proc.\ 23rd ACM-SIAM Symp.\ on Discrete Algorithms}, pages
  1150--1162, 2012.

\bibitem{holzer14}
S.~Holzer and N.~Pinsker.
\newblock {Approximation of Distances and Shortest Paths in the Broadcast
  Congest Clique}.
\newblock {\em CoRR}, abs/1412.3445, 2014.

\bibitem{lenzen13}
C.~Lenzen and D.~Peleg.
\newblock Efficient distributed source detection with limited bandwidth.
\newblock In {\em Proc.\ 32nd ACM Symp.\ on Principles of Distributed
  Computing}, 2013.

\bibitem{ARPANET}
J.~McQuillan, I.~Richer, and E.~Rosen.
\newblock The new routing algorithm for the {ARPANET}.
\newblock {\em IEEE Trans. Communication}, 28(5):711--719, May 1980.

\bibitem{MQ-77}
J.~M. McQuillan and D.~C. Walden.
\newblock The {ARPANET} design decisions.
\newblock {\em Networks}, 1, 1977.

\bibitem{ospf}
J.~Moy.
\newblock {OSPF} version 2, April 1998.
\newblock Internet RFC 2328.

\bibitem{nanongkai14}
D.~Nanongkai.
\newblock {Distributed Approximation Algorithms for Weighted Shortest Paths}.
\newblock In {\em Proc.\ 46th Symposium on Theory of Computing (STOC)}, pages
  565--573, 2014.

\bibitem{lenzen12}
B.~Patt-Shamir and C.~Lenzen.
\newblock {Fast Routing Table Construction Using Small Messages [Extended
  Abstract]}.
\newblock In {\em Proc. 45th Symposium on the Theory of Computing (STOC)},
  2013.
\newblock {Full version at \url{http://arxiv.org/abs/1210.5774}.}

\bibitem{Peleg:book}
D.~Peleg.
\newblock {\em Distributed Computing: A Locality-Sensitive Approach}.
\newblock SIAM, Philadelphia, PA, 2000.

\bibitem{peleg89}
D.~Peleg and A.~A. Sch\"affer.
\newblock Graph spanners.
\newblock {\em J. Graph Theory}, 13(1):99--116, 1989.

\bibitem{peleg89trade-off}
D.~Peleg and E.~Upfal.
\newblock A trade-off between space and efficiency for routing tables.
\newblock {\em Journal of the ACM}, 36(3):510--530, 1989.

\bibitem{PU89}
D.~Peleg and E.~Upfal.
\newblock A trade-off between space and efficiency for routing tables.
\newblock {\em J. ACM}, 36(3):510--530, 1989.

\bibitem{thorup01}
M.~Thorup and U.~Zwick.
\newblock Compact routing schemes.
\newblock In {\em Proc.\ 13th ACM Symp. on Parallel Algorithms and
  Architectures}, 2001.

\bibitem{zwick01}
U.~Zwick.
\newblock Exact and approximate distances in graphs - a survey.
\newblock In {\em Proc.\ 9th European Symp.\ on Algorithms}, pages 33--48,
  2001.

\bibitem{zwick02}
U.~Zwick.
\newblock All pairs shortest paths using bridging sets and rectangular matrix
  multiplication.
\newblock {\em Journal of the {ACM}}, 49(3):289--317, 2002.

\end{thebibliography}

\end{document}